\documentclass[10pt,a4paper,twocolumn]{article}

\usepackage[utf8]{inputenc}
\usepackage[T1]{fontenc}

\usepackage{floatrow}
\usepackage{cite}
\usepackage{paralist}
\usepackage{algorithmic}
\usepackage[pdftex]{graphicx}
\usepackage[cmex10]{amsmath}
\usepackage{amssymb}
\usepackage{amsthm}
\usepackage{textcomp}
\usepackage{mathptmx}
\usepackage{amsfonts}
\usepackage{booktabs}
\usepackage{array}
\usepackage{url}
\usepackage{hyperref}
\usepackage{mathrsfs}
\usepackage{wrapfig}
\usepackage[caption=false,font=footnotesize]{subfig}
\usepackage[vmargin=25mm,hmargin=15mm]{geometry}
\usepackage{fixltx2e}

\usepackage{microtype}
\usepackage[disable]{todonotes}

\theoremstyle{theorem}
\newtheorem{theorem}{Theorem}

\theoremstyle{example}
\newtheorem{example}{Example}

\theoremstyle{definition}
\newtheorem{remark}{Remark}

\newcommand{\mnset}{\ensuremath{\mathcal{M}}}
\newcommand{\nnset}{\ensuremath{\mathcal{N}}}
\newcommand{\qset}{\ensuremath{\mathcal{Q}}}
\newcommand{\mf}{\enspace .}
\newcommand{\mc}{\enspace ,}
\DeclareMathOperator{\ii}{int}
\DeclareMathOperator{\ee}{ext}
\newcommand{\mil}{\ensuremath{\mu}-link}
\newcommand{\nol}{\ensuremath{\nu}-link}
\newcommand{\vsa}{VS1}
\newcommand{\vsb}{VS2}
\newcommand{\mils}{\mil{}s}
\newcommand{\nols}{\nol{}s}
\newcommand{\mn}{\ensuremath{\mu\nu}}
\newcommand{\nn}{\ensuremath{\nu\nu}}
\newcommand{\deq}{\mathrel{\mathop:}=}

\newlength{\plen}
\begin{document}

\title{OutFlank Routing: Increasing Throughput in Toroidal Interconnection
  Networks\footnote{To appear in the 19th IEEE International Conference on
    Parallel and Distributed Systems (ICPADS 2013).}}

\author{Francesco Versaci}
\date{Vienna University of Technology}

\maketitle

\begin{abstract}
  We present a new, deadlock-free, routing scheme for toroidal interconnection
  networks, called OutFlank Routing (OFR). OFR is an adaptive strategy which
  exploits non-minimal links, both in the source and in the destination nodes.
  When minimal links are congested, OFR deroutes packets to carefully chosen
  intermediate destinations, in order to obtain travel paths which are only an
  additive constant longer than the shortest ones.  Since routing performance is
  very sensitive to changes in the traffic model or in the router parameters, an
  accurate discrete-event simulator of the toroidal network has been developed
  to empirically validate OFR, by comparing it against other relevant routing
  strategies, over a range of typical real-world traffic patterns. On the
  16\texttimes16\texttimes16 (4096 nodes) simulated network OFR exhibits
  improvements of the maximum sustained throughput between 14\% and 114\%, with
  respect to Adaptive Bubble Routing.
\end{abstract}\\


\section{Introduction}

Tori are key topologies for the interconnection networks of parallel computers,
both in multiprocessors machines and on-chip multiprocessors
\cite{DallyT01}. E.g., the Fujitsu K computer features a six-dimensional torus
network \cite{AjimaTIHS11}, IBM BlueGene/Q a five-dimensional one
\cite{ChenEtal11}, while both the Cray XT series \cite{BrooksK11} and IBM
BlueGene/L and P \cite{adiga2005blue,alam2008bluegenp} have a three-dimensional
one. Furthermore, the adoption of mesh-like topologies is bound to increase even
more, since they represent the only scalable solution under a limiting
technology point of view \cite{BilardiP95}.

Routing is the fundamental problem of choosing paths (or other network
resources, such as link queues) for packets which travel in a network, between a
source node $s$ and a destination node $t$, and has been studied intensively,
both by theoretical means and empirical investigations. Typically, routing is
studied under some permutation pattern, which causes each node to send and
receive the same amount of data.
Routing algorithms can be roughly classified in \emph{oblivious}, for which the
path from $s$ to $t$ depends only on $s$ and $t$ (and optionally on some
randomness), and \emph{adaptive}, for which decisions on the path can also
depend on online network parameters, such as, e.g., congestion at edges.

Theoretical investigations have primarily focused on oblivious algorithms, partly
because of the difficulties encountered in analyzing adaptive ones, exemplified
in 1982 by L.~Valiant \cite{Valiant82}: ``Unfortunately such [adaptive]
strategies appear to be beyond rigorous analysis or testability''. The first
nontrivial lower bounds concerning a wide class of adaptive algorithms for the
two-dimensional mesh have been presented in \cite{ChinnLT96,Ben-AroyaCS98} (and,
furthermore, are the first to take into account bounded queues in the routers).
Among oblivious strategies there is a large performance gap if randomness is
allowed: e.g., deterministic policies perform far from optimal under a
worst-case congestion metrics, since on a graph with $N$ nodes and maximum
degree $d$ there always exists a permutation which causes edge congestion
$\Omega\left(\sqrt{N}/d\right)$ \cite{BorodinH85,KaklamanisKT91} while, on the
other hand, randomized strategies can achieve a $O(\log N)$ competitive ratio
against the optimal offline
\cite{ValiantB81,BuschMX08,Racke09,Racke08,Vocking01}.
However, in current supercomputers randomized oblivious strategies have seldom
been implemented, a notable exception being the recent PERCS interconnect
\cite{ArimilliEtal10}.

Though difficult to analyze, adaptive routing algorithms have shown to perform
remarkably well on real machines and have thus been adopted, e.g., by IBM
BlueGene/L, P and Q \cite{BlumrichEtal03,alam2008bluegenp,ChenEtal11}, Cray T3E
\cite{scott1996cray} and Quadrics \cite{GeoffrayH08}.
An important property of adaptive algorithms is whether they are \emph{minimal},
i.e., if they allow only shortest paths from source to destination (if
\emph{any} shortest path is admissible then they are called \emph{fully}
minimal).
Minimal strategies ensure shortest delivery time when the network is under-loaded
and preclude \emph{livelock}, i.e., the possibility for a packet to move
indefinitely in the network without reaching its destination.
On the other hand, minimal strategies underutilize the network bandwidth: e.g.,
on a three-dimensional torus there are six outgoing links from each node, but
only three (or less) of these lie along minimal directions to a given
destination node; hence, the maximum throughput achievable in the torus by a
minimal routing can be half (or less) of the one achievable by allowing non
minimal paths.

An interesting technique for increasing the maximum throughput while keeping the
router complexity low has been proposed by Singh et al.\
\cite{SinghDTG02,SinghDGT03,SinghDTG04,SinghDGT04} for $k$-ary $n$-cubes (i.e.,
for $n$-dimensional tori with length $k$ in each dimension).  They have
suggested to deroute packets non-minimally along ``crowded'' dimensions, in
order to decrease the congestion of the \emph{minimal orthant}, i.e., the part
of the network which contains all (and only) the minimal paths.
In more detail, to route each \emph{newly injected} packet, the router can
choose among $2^n$ orthants, based on the congestion of the outgoing links. Once
the orthant has been decided (e.g., $(x+, y-)$ in two dimensions), the packet
can travel in each dimension only according to the associated directions.  In
the following, we refer to such a general strategy (which comprises the GAL,
GOAL and CQR policies) as Pick Orthant Routing (POR).
This kind of derouting has the clear advantage of spreading the traffic along
the network, hence decreasing the congestion in hot spots, but can arbitrarily
lengthen a path, since it may force a packet, which has as destination a
neighbor of the source node, to move in the non-minimal direction all around the
torus (thus traversing $k-1$ nodes instead of just one).

In this work, we propose an adaptive strategy which tries to make use of the
full bandwidth while also limiting the path dilation to an \emph{additive
  constant}.
The OutFlank Routing (OFR) schematically works as follows: when the minimal
links are congested it starts derouting packets to Intermediate Destination
Nodes (IDNs). There are two types of IDNs which are considered: Wraparound IDNs,
which mimic the POR routing strategy, and OutFlank IDNs, which \emph{outflank}
the minimal orthant and only increase the path length by a constant value,
while still exploiting the available non-minimal links, both in the source and
the destination nodes (see Fig.~\ref{fig:ofscheme} for an example in two
dimensions). Another interesting feature of the OFR policy comes from the fact
that, by using only OutFlank IDNs, it might also be adapted to work on
(non-toroidal) mesh networks, whereas the POR strategy is strictly linked to the
toroidal topology.  In this paper, however, we will consider only toroidal
networks, in order to have a more uniform comparison with related routing
strategies.
Since routing performance is very sensitive to changes in the traffic generation
model or in the parameters of the adopted routing strategy, we have developed an
accurate discrete-event simulator of the toroidal networks (built on top of the
OMNeT++ network simulation framework \cite{Varga10}), which we have used to
empirically validate OFR. In detail, we have compared it against both POR and
the Adaptive Bubble Router \cite{PuenteBGPDI99,PuenteIBGVP01} (ABR, described in
Sec.~\ref{sec:rout-toro-netw}), adopted in BlueGene machines, for various
real-world traffic patterns.

\begin{figure*}
  \centering
  \subfloat[OutFlank IDNs -- $q^0$ and $q^1$ are OIDNs]{
    {\includegraphics[width=89mm]{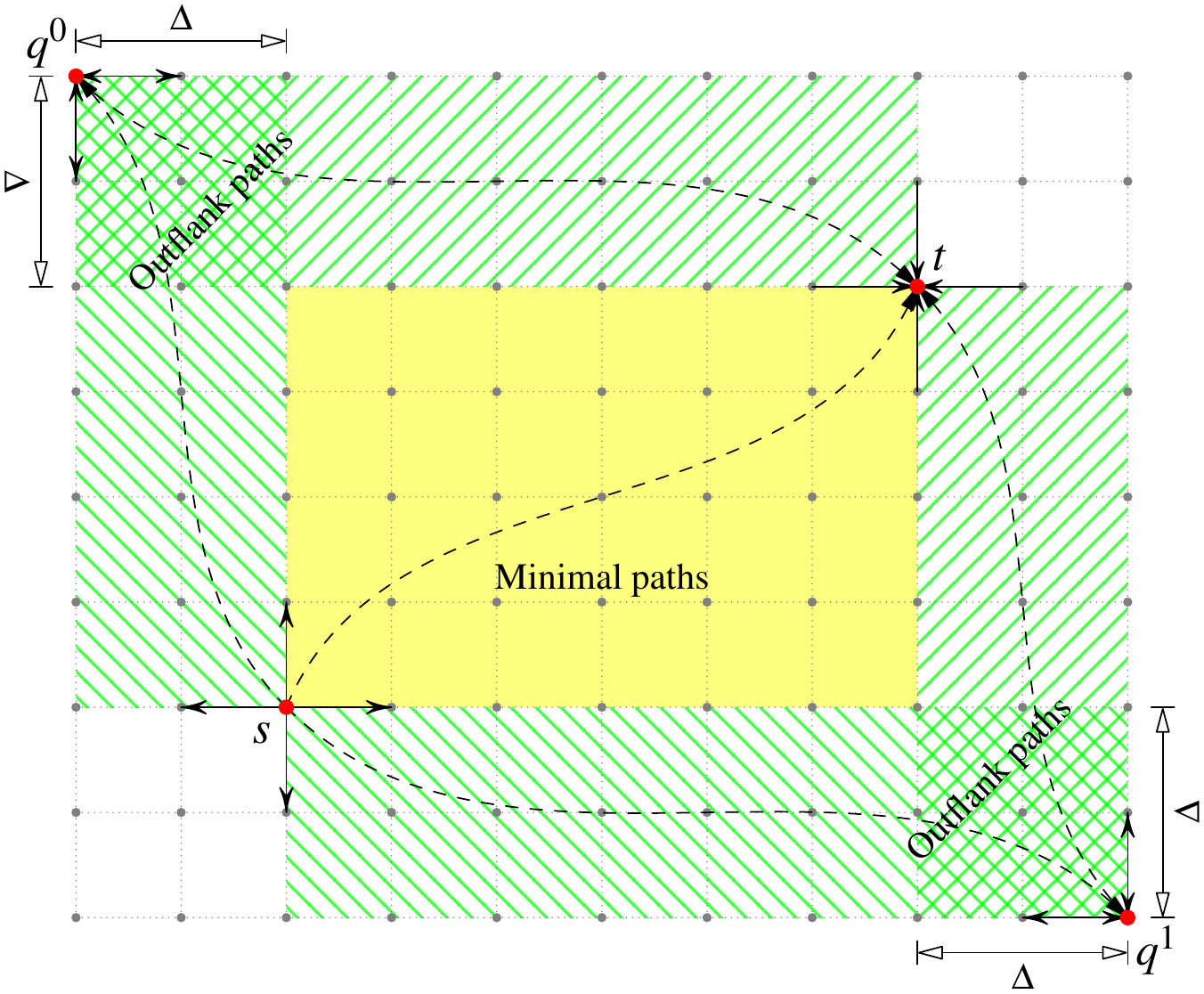}}
    \label{fig:of2da}
  }
  \hspace{5mm}
  \subfloat[Wraparound IDNs -- $q^0$, $q^1$ and $q^2$ are WIDNs]{
    {\includegraphics[width=80mm]{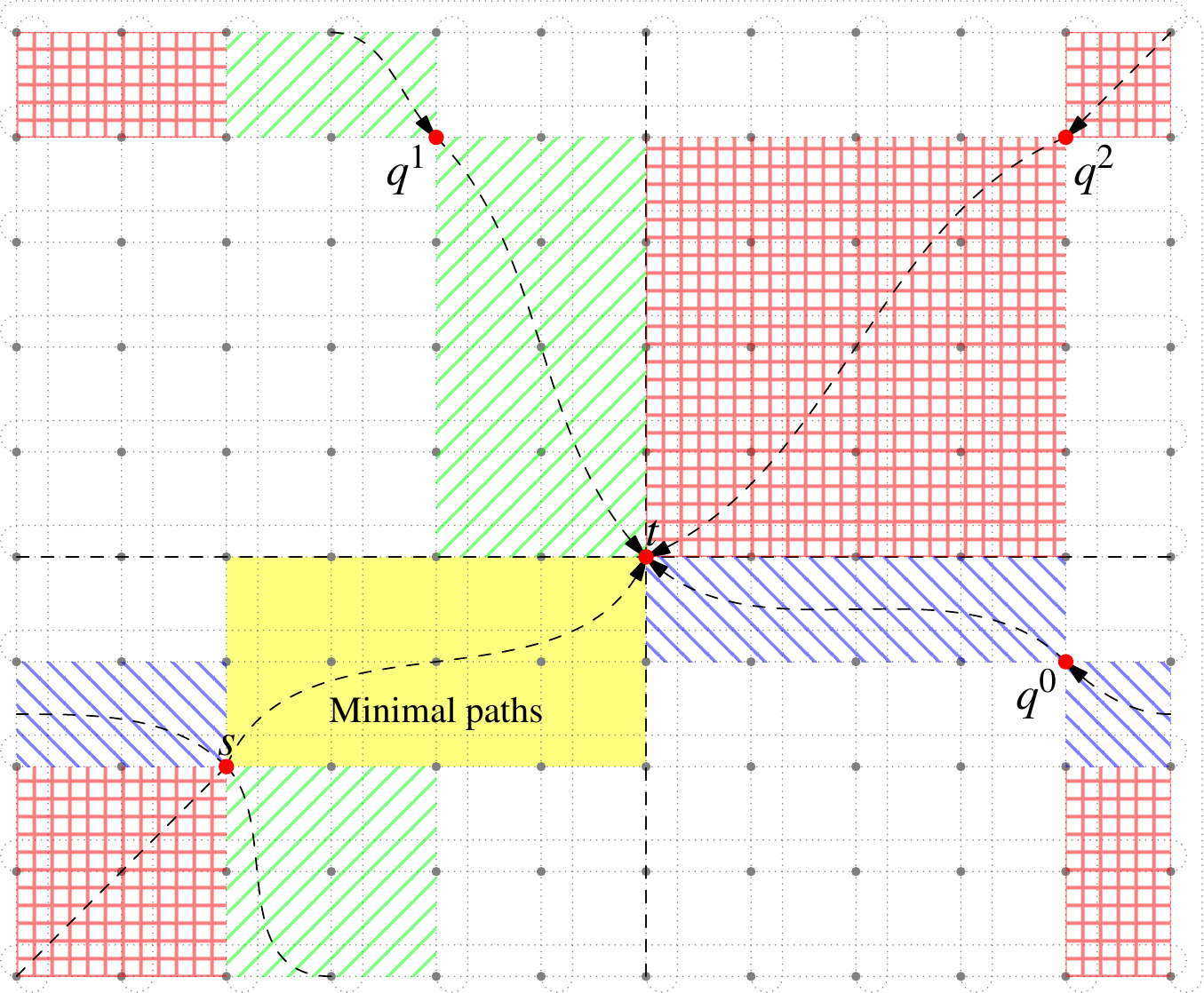}}
    \label{fig:of2db}
  }
  \caption{OutFlank routing in the two-dimensional torus.}
  \label{fig:ofscheme}
\end{figure*}

The paper is structured as follows: in Section~\ref{sec:rout-toro-netw} we
present the model adopted for the network and the traffic generation and review
some relevant results which will be used later; in
Section~\ref{sec:outfl-rout-algor} we illustrate OFR for general toroidal
topologies, giving a detailed description of the widespread two and
three-dimensional cases; in Section~\ref{sec:exper-valid} we describe the
discrete-event simulator and discuss the experimental results obtained; finally,
Section~\ref{sec:concl-future-work} contains the conclusions.

\section{Network model and preliminaries}
\label{sec:rout-toro-netw}

We consider $n$-dimensional toroidal interconnection networks, with
bidirectional links between the nodes. Each physical node comprises of
\begin{enumerate}
\item a generator which produces the packets to be sent,
\item a router which routes the packets,
\item a sink which consumes the packets.
\end{enumerate}
Traffic is generated in bunches of $M$ packets, all having the same destination
(a bunch represents a \emph{message}), which are injected into the router via a
dedicated link with bandwidth $B_{\ii}$ and latency $\lambda_{\ii}$. A similar
link connects the router to the sink. The router has outgoing links to its $2n$
neighbors, with bandwidth $B_{\ee}$ and latency $\lambda_{\ee}$. Typically, the
packets size $S$ is chosen to have $\lambda_{\ee}$ and $\frac{S}{B_{\ee}}$ of
the same order of magnitude, so that
\begin{inparaenum}[(i)]
\item programs which perform small communications do not suffer large latencies
  due to a large packet size,
\item programs which communicate with large chunks of data do not suffer too
  much overhead from fragmentation of the messages.
\end{inparaenum}

Switching of packets in routers is generally categorized in:
\begin{description}
\item[\rm\em Store-and-forward:] Messages are split in packets and sent
  \emph{independently}. Each packet is received, stored completely in a router
  queue and then forwarded.
\item[\rm\em Wormhole:] Messages are divided into flits, which are \emph{orderly} routed
  along the \emph{same} path \cite{DallyS87}. Routers can forward a message if
  they have storage space for at least one flit.
\item[\rm\em Virtual cut-through:] Messages are split in packets, sent
  independently. Packets are forwarded only if there is enough space to store
  them completely. As soon as the packet header is read the whole packet is
  \emph{simultaneously} buffered locally and forwarded to an outgoing link
  \cite{KermaniK79}.
\end{description}
The OFR paradigm applies to all of the above switching techniques, but the
current simulator employed in Sec.~\ref{sec:exper-valid} implements a
store-and-forward approach.

When designing a routing strategy a main concern is to avoid \emph{deadlocks},
i.e., presence of packets permanently stuck because of circular dependencies in
the resource requests. The theory of deadlock
\cite{DallyS85,DallyS87,Duato93,Duato96} typically makes use of various virtual
channels to increase the available links, so as to be able to impose a priority
order in the resources usage and remove the possibility of circular
dependencies.
A different, simpler approach \cite{Roscoe87} has been adopted by ABR
\cite{PuenteBGPDI99,PuenteIBGVP01}, which works by preventing all buffers to get
completely filled: if some free slots (\emph{bubbles}) are always left in the
network, then packets can constantly move and deadlock cannot occur.
In detail, ABR has two Virtual Channels (VCs) for each unidirectional link, an
\emph{escape} (bubble) one and an \emph{adaptive} one. Packets in the escape
network move in dimension order and can be injected into a router (from the
generator or from a lower dimension) only if there are at least \emph{two} slots
available.  (The dimension order breaks deadlocks between packets moving along
different dimensions, the bubble rule breaks deadlocks inside one-dimensional
cycles.)
The adaptive subnetwork complements the escape one: packets are preferably
routed into adaptive queues, if available; any adaptive queue lying along a
minimal path can be chosen (fully minimal) and only when all the adaptive queues
are full, escape queues are considered (packets moving from adaptive to escape
queues are regarded as injection, i.e., two free slots are required for the
insertion).

\section{The OFR algorithm}
\label{sec:outfl-rout-algor}

OFR can be regarded as a higher level mechanism which exploits an underlying
minimal router: when a packet is firstly injected into the network, OFR decides
if it is convenient to deroute it and, if that is the case, chooses a suitable
IDN $q$ for it; packet header is extended to include the optional IDN and
packets are first routed (\emph{minimally}) to their IDN (if they have one) and
then to their final destination. We have chosen to adopt ABR as the underlying
minimal routing scheme for OFR. The main task of OFR is to carefully choose if
and where to deroute a packet.  In this section we thoroughly describe these
mechanisms.
In \S\ref{sec:admiss-interm-points} we define the set of admissible IDNs which
can be used for derouting; in \S\ref{sec:how-choose-interm} we see how to
adaptively choose among the possible IDNs, weighting congestion and path
dilation; in \S\ref{sec:deadlock-livelock} we discuss how to prevent deadlock
and livelock to occur when using OFR.

\subsection{Admissible IDNs}
\label{sec:admiss-interm-points}

Consider an $n$-dimensional toroidal network of size $k_1\times \cdots\times
k_n$, a source node $s$ with coordinates $s=(s_1, \ldots, s_n)$ and a
destination node $t=(t_1, \ldots, t_n)$. There are $2n$ outgoing links from $s$
and $2n$ incoming links to $t$, providing an upperbound on the exploitable
bandwidth between $s$ and $t$.  Each of the above links can be minimal (\mil) or
non-minimal (\nol). Each dimension, under the point of view of the source $s$,
$i\in\{1, \ldots, n\}$, can be classified either as:
\begin{enumerate}
\item \mn\ (minimal/non-minimal), if $s_i \not = t_i$,
\item \nn\ (non-minimal/non-minimal), if $s_i = t_i$.
\end{enumerate}
We call \mnset\ the set of all \mn\ dimensions, and \nnset\ the set of all \nn\
dimensions.  (If $i\in\mnset$, then one of the two outgoing links from $s$ along
dimension $i$ is a \mil, and the other one is a \nol; if $i\in\nnset$, then both
the two outgoing links from $s$ along dimension $i$ are \nols.)
Minimal routers use only \mils, both in the source and in the destination
nodes, and hence cannot exploit \nn\ dimensions. OFR strategy takes advantage
also of \nols, both along \mn\ and \nn\ dimensions.

OFR exploits non-minimal links by considering two kinds of possible IDNs:
Wraparound IDNs (WIDNs) and OutFlank IDNs (OIDNs).

\subsubsection{Wraparound IDNs}
\label{sec:wraparound-idns}

WIDNs enable OFR to mimic POR behavior.
In more detail, let $\beta$ be an $n$-dimensional binary vector:
$\beta=(\beta_1, \ldots, \beta_n)\in\{0,1\}^n$, then an admissible WIDN
$q(\beta)=(q_1(\beta_1), \ldots, q_n(\beta_n))$, is defined by setting
\begin{align}
  q_i(\beta_i) = \left\lfloor \frac{s_i+t_i+\beta_i k_i}{2} \right\rfloor \mf
\end{align}

When $\beta$ varies over all the possible $2^n$ values, the associated WIDNs lie
in the middle of each of the $2^n$ orthants (including the minimal one, which
can be discarded since it would not deroute the packet). By forcing a packet to
go through a WIDN we make it wrap around some edges, thus simulating the POR
behavior (see Fig.~\ref{fig:of2db} for an example in the two-dimensional case).

\subsubsection{OutFlank IDNs}
\label{sec:outflank-idns}

OIDNs are the IDNs which allow OFR to achieve constant dilation by
\emph{locally} derouting a packet.  Intuitively speaking, the OIDNs are chosen
near the vertices of the minimal orthant and outside of it (e.g., $\Delta$ steps
outside in each dimension). Formally, OIDN $q(\lambda)=(q_1(\lambda_1), \ldots,
q_n(\lambda_n))$ is determined by choosing for each dimension $i$ a value
$\lambda_i\in\{-1, 0, 1\}$ s.t.{}
\begin{align}
  \forall i\in\mnset \nonumber \\ q_i(\lambda_i) &=
  \begin{cases}
    ( s_i - \Delta ) \,\bmod k_i\mc & \lambda_i=-1 \quad\wedge\quad \text{$i-$ is a \nol} \\
    ( t_i - \Delta ) \,\bmod k_i\mc & \lambda_i=-1 \quad\wedge\quad \text{$i-$ is a \mil} \\
    ( s_i + \Delta ) \,\bmod k_i\mc & \lambda_i=1 \quad\wedge\quad \text{$i+$ is a \nol} \\
    ( t_i + \Delta ) \,\bmod k_i\mc & \lambda_i=1 \quad\wedge\quad \text{$i+$ is a \mil} \\
    \left\lfloor\frac{s_i+t_i}{2}\right\rfloor \mc & \lambda_i=0
  \end{cases}\mc\\
  \forall i\in\nnset \nonumber \\
  q_i(\lambda_i) &=
  \begin{cases}
    s_i + \lambda_i \Delta \, \bmod k_i\mc & \lambda_i \in \{-1,1\} \\
    s_i\mc & \lambda_i=0
  \end{cases}\mc
\end{align}
where $\Delta$ is a router parameter ($\Delta=2$ in the simulations of
Sec.~\ref{sec:exper-valid}).
Furthermore, OIDNs are chosen so as to exploit at least a \nol\ both in $s$ and
in $t$.
Let us now see an example of the admissible OIDNs to help clarify the general
procedure.

\begin{figure}
  \centering
  \includegraphics[width=90mm]{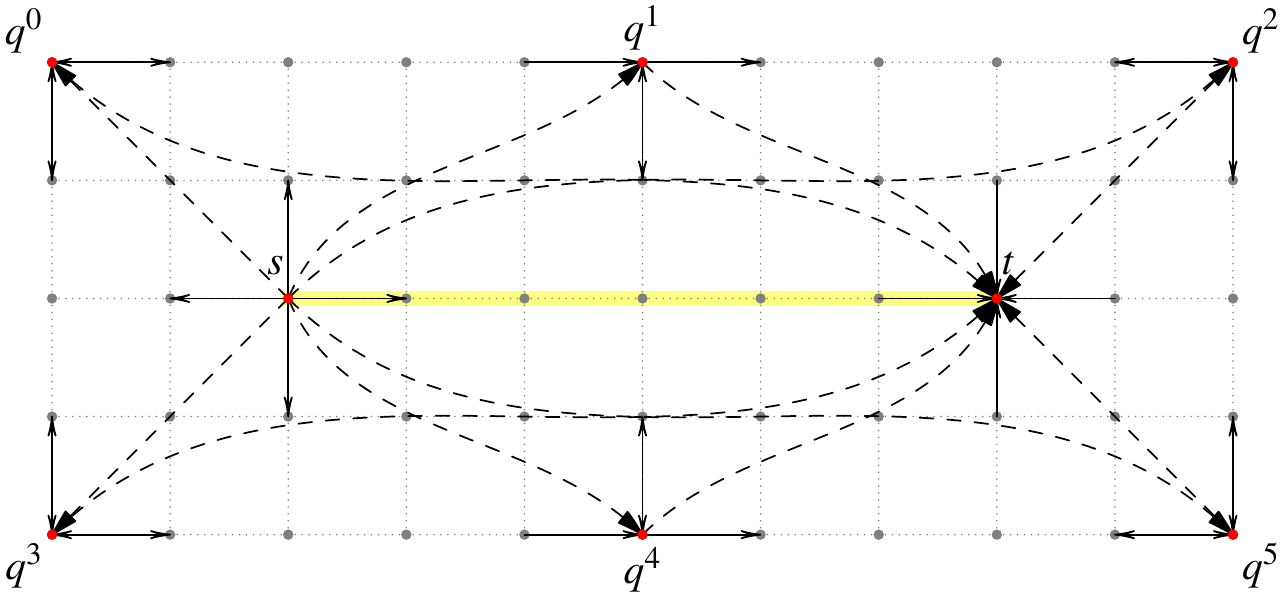}
  \caption{OFR -- Admissible OIDNs for collinear $s$ and $t$ (two-dimensional
    case).}
  \label{fig:inline-2d-of}
\end{figure}

\begin{example}[OFR in two-dimensional tori]
  \label{example:OFR-2dim}
  Consider Fig.~\ref{fig:of2da}: $s$ and $t$ have different coordinates in
  both the dimensions, so both $1$ (i.e., $x$) and $2$ (i.e., $y$) belong to
  \mnset. The minimal orthant is the yellow area, and the \mils\ are $(x+,
  y+)$. Since we want to use both a \mil\ and a \nol, the two possible choices
  are $(x-, y+)$ and $(x+, y-)$, which correspond to OIDNs $q^0$ and $q^1$ and
  $\lambda$ values $(1,-1)$ and $(-1,1)$. Now consider the case, illustrated in
  Fig.~\ref{fig:inline-2d-of} of $s$ and $t$ collinear. In this case
  $x\in\mnset$ and $y\in\nnset$, so the admissible choices for $\lambda$ are:
  \begin{align}
    \lambda^0&=(-1,1) \mc  & \lambda^1&=(0,1) \mc  & \lambda^2&=(1,1)  \mc \\
    \lambda^3&=(-1,-1)\mc & \lambda^4&=(0,-1)\mc & \lambda^5&=(1,-1) \mc
  \end{align}
  which induce the six OIDNs $q^0, \ldots, q^5$ represented in
  Fig.~\ref{fig:inline-2d-of}.
\end{example}

\begin{remark}
  When designing the actual router, not all the admissible OIDNs need to be
  considered for derouting, but it suffices to choose a subset of them which
  covers all the \nols, both in the source and destination nodes. E.g., in the
  case of collinear $s$ and $t$ of Example~\ref{example:OFR-2dim}, the four OIDNs
  $q^0$, $q^2$, $q^3$, and $q^5$ are enough for enabling full bandwidth between
  $s$ and $t$.
\end{remark}
\begin{remark}
  If for some dimension $i$ the distance between $s_i$ and $t_i$ is not smaller
  than $\frac{k_n}{2} - \Delta$, then the computed OIDNs can ``jump'' into
  another orthant because of the torus wraparound. The above procedure still
  works correctly, but the OIDNs obtained now behave, for some dimensions,
  similarly to WIDNs (because of the wraparound).
\end{remark}

Outflanking through an OIDN causes a path dilation: the total derouted distance
$\tilde d(s,t)$ is larger than the minimal distance $d(s,t)$. In detail, using
an OIDN $q$ yields an additional path distance $\delta(q)$, but this
extra-distance is always bounded by a constant (once $n$ and $\Delta$ are
fixed):
\begin{align}
  \tilde d_q(s,t) &\deq d(s,q)+d(q,t) = d(s,t) + \delta(q) \mc & \forall q \; \delta(q) &\leq 2n\Delta \mf  
\end{align}
We have set $\Delta$ to a constant, but one might also consider to set $\Delta$
to be a fraction of $d(s,t)$, so as to increase the outflank throughput while
keeping the dilation bounded by a \emph{multiplicative} factor.
Focusing again on Fig.~\ref{fig:inline-2d-of}, the additive dilations of the
admissible OIDNs are
\begin{align}
  \delta^0&=4\Delta \mc & \delta^1&=2\Delta \mc & \delta^2&=4\Delta \mc\\
  \delta^3&=4\Delta \mc & \delta^4&=2\Delta \mc & \delta^5&=4\Delta \mf
\end{align}
Using the dilations above we can choose another subset of the admissible OIDNs
which still covers all the \nols\ as before, but induces smaller dilations
(e.g.: $q^0$, $q^1$, $q^4$, and $q^5$).

\begin{example}[OFR in three-dimensional tori]
  We list now the OIDNs used by the OFR router for the three-dimensional case
  (the one implemented in the simulator). We assume, w.l.o.g., that if a \mil\
  exists along some dimension $i$, then it is positively oriented ($i+$).

  \begin{itemize}
  \item $s_1\not= t_1, s_2\not= t_2, s_3\not= t_3$ ($s$ and $t$ not coplanar)
    \begin{align}
      \lambda^0&=(0,-1,1) \mc & \lambda^1&=(0,1,-1) \mc & \lambda^2&=(-1,0,1) \mc \\
      \lambda^3&=(1,0,-1) \mc & \lambda^4&=(-1,1,0) \mc & \lambda^5&=(1,-1,0)
      \mf
    \end{align}
  \item $s_1= t_1, s_2\not= t_2, s_3\not= t_3$ ($s$ and $t$ coplanar but not
    collinear, see Fig.~\ref{fig:coplanar-3d-ofr})
    \begin{align}
      \lambda^0&=(0,-1,1) \mc & \lambda^1&=(0,1,-1) \mc \\
      \lambda^2&=(1,0,0) \mc & \lambda^3&=(-1,0,0) \mf
    \end{align}
  \item $s_1= t_1, s_2= t_2, s_3\not= t_3$ ($s$ and $t$ collinear)
    \begin{align}
      \lambda^0&=(1,0,1) \mc & \lambda^1&=(-1,0,0) \mc \\
      \lambda^2&=(0,1,-1) \mc & \lambda^3&=(0,-1,0) \mf
    \end{align}
  \end{itemize}
\end{example}

\begin{figure*}
  \centering
  \includegraphics[width=14cm]{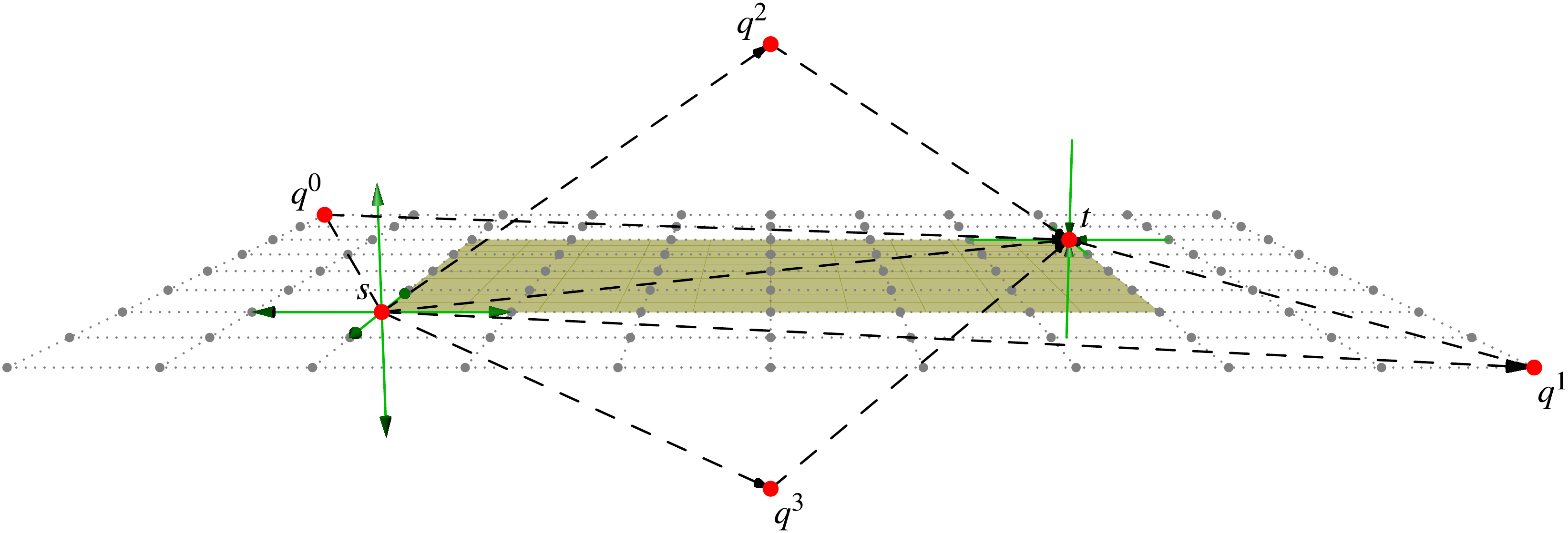}
  \caption{The four OIDNs used by OFR for coplanar $s$ and $t$ (three-dimensional
    case).}
  \label{fig:coplanar-3d-ofr}
\end{figure*}

\subsection{Choosing the IDN}
\label{sec:how-choose-interm}

To choose the most rewarding intermediate destination, OFR computes a profit for
each available IDN, taking into account both the dilation and the congestion of
the alternative route.
In more detail, OFR works as follows:
\begin{enumerate}
\item When a packet is newly injected, the router computes a set \qset\ of
  IDNs that covers all the \nols\ from $s$ to $t$.
\item For each $q\in\qset$ two values are computed:
  \begin{enumerate}
  \item The total distance $\tilde d_q(s,t)=d(s,q)+d(q,t)$.
  \item The average number $u_q(s)$ of used (non free) slots in the (local)
    outgoing links which are minimal from $s$ to~$q$.
  \end{enumerate}
\item We obtain a single profit, by scalarizing the two metrics in which we are
  interested, as
  \begin{equation}
    \label{eq:profit}
    \pi_q(s,t) \deq \frac{u_*(s)}{u_q(s,t)}+\eta\frac{d(s,t)}{\tilde d_q(s,t)} \mc
  \end{equation}
  where $u_*(s)$ is the minimum number of used slots in a link going out from
  $s$ and $\eta$ is parameter tuned empirically ($\eta=2.0$ in the simulations
  of Sec.~\ref{sec:exper-valid}). Also, the first fraction is set to one
  whenever $u_q(s,t)=u_*(s)=0$.
\item The analogous profit is computed also for the minimal (not outflanked)
  routing:
  \begin{equation}
    \pi_0(s,t) \deq \frac{u_*(s)}{u_0(s,t)}+\eta \mc
  \end{equation}
  where $u_0(s)$ is the average number of used slots in \mils\ from $s$.
\item If, for all $q$, $\pi_0(s,t) \geq \pi_q(s,t)$, then the packet is routed
  minimally and no derouting is performed.
\item Otherwise, $q^*$ is chosen as IDN, such as
  \begin{equation}
    q^* \deq \arg\max_q \pi_q(s,t) \mf
  \end{equation}
\end{enumerate}

\begin{remark}
  The same procedure has been used in our simulations to decide when to deroute
  packets while adopting POR, but the parameter $\eta$ has been independently
  optimized ($\eta=1.0$ for POR).
\end{remark}

\subsection{Deadlock and livelock}
\label{sec:deadlock-livelock}

In order to be usable, any routing policy needs to be deadlock-free, or to
include some mechanism to detect deadlocks and recover from them
\cite{AnjanP95,Martinez-RubioLD03}.  OFR, being a non-minimal strategy, might
also be subject to livelocks. 
Actually, if not carefully implemented, OFR can produce deadlocks: if a packet
reaches an IDN and the router simply tries to reroute it towards its final
destination, then a loop of dependencies might arise from other packets, since
derouting interferes with ABR dimension order routing in the escape subnetwork.

The following two theorems prove that OFR can be implemented without deadlocks
and livelocks.

\begin{theorem}
  OFR can be implemented in a deadlock-free way by adopting VCs in the escape
  network.
\end{theorem}
\begin{proof}
  OFR can interfere with ABR mechanisms to prevent deadlocks in the following
  way: when a packet reaches its intermediate destination and it is rerouted
  toward its final destination, it might come from an ABR escape VC and try to
  enter into another escape VC which is lower in the dimension order used by
  ABR. This could lead to cyclic dependencies and hence deadlock in OFR.
  To avoid this, it is enough to split the escape subnetwork in two virtual
  escape subnetworks \vsa\ and \vsb. \vsa\ is used only by packets which are
  reaching an IDN, whereas \vsb\ is used only by packets routed towards their
  final destination (note that we only need to split the escape subnetwork, not
  the adaptive one). An insertion into \vsb\ of a packet which has reached its
  IDN is treated as an injection, i.e., the insertion is allowed only if it does
  not saturate the queue. Packets which have entered \vsb\ will always be
  delivered, since \vsb\ is deadlock-free and does not depend on other external
  resources. Therefore, some resource will be freed up for packets which want to
  enter in \vsb\ from \vsa\ (or directly from the generators if no IDN is given
  for a packet), propagating the deadlock-freedom to the whole network.
\end{proof}

\begin{theorem}
  OFR is livelock-free.
\end{theorem}
\begin{proof}
  OFR deroutes packets by transforming a path $s \rightarrow t$ into a longer
  path $s \rightarrow q \rightarrow t$. Since derouting can only happen to newly
  injected packets (i.e., once in a packet lifetime) there cannot be path
  decompositions with more than one intermediate destination. Since each of the
  two travel phases $s \rightarrow q$ and $q \rightarrow t$ are routed
  minimally, livelock cannot occur.
\end{proof}

\section{Experimental validation}
\label{sec:exper-valid}

Routing performance is extremely sensitive to implementation details, network
architecture and algorithmic parameters. To get a realistic estimate of OFR
behavior we have therefore chosen to develop an accurate network simulator.
Our discrete event simulator has been written within the OMNeT++ network
simulation framework \cite{Varga10}; the simulator code will be released under
the GNU Lesser General Public License v3.0 \cite{lgpl}. The implemented network
is inspired by the three-dimensional toroidal network being developed within the
AuroraScience project \cite{PivantiSS11}.

We have tested the network performance of OFR under a sustained traffic and with
different communication patterns, and we have compared OFR with both ABR and POR
routing schemes. To make the comparison as uniform as possible we have
implemented POR (as well as OFR) in an ABR-like fashion and we have chosen when
to deroute packets for POR using the technique described in
\S\ref{sec:how-choose-interm} for OFR, but optimized independently for
maximizing POR performance.

Traffic is generated in bursts (messages) of $M$ packets, according to the
chosen communication pattern.  The offered load $\gamma$, the rate at which
packets are generated, is normalized as follows: it is equal to $\gamma_0=1$
when it saturates the bisection bandwidth of the network under uniform traffic
(i.e., $\gamma_0=1$ is an upper bound to the sustainable injection rate if the
packets are sent to uniformly random destinations). Depending on the
communication pattern $P$ and routing algorithm $R$ there is a value
$\gamma^*(R,P)$ at which the network is saturated and cannot sustain the traffic
(i.e., the average lifetime of the packets grows to infinity).  We are
interested in both the maximum sustained throughput $\gamma^*$ and in the average
latency of packets for fixed values of the offered load $\gamma$.
In the plots of Fig.~\ref{fig:plots-16} we show, for each simulated
communication pattern on the $16 \times 16 \time 16$ network, the average packet
lifetime as a function of $\gamma$.  The times are plotted only for values of
$\gamma$ which did not saturate the network (i.e., times not reported are
infinite). A summary of the achieved maximum throughputs $\gamma^*$, for the tested
routing algorithms and communication patterns, both for $k=8$ and $k=16$, is
given in Table~\ref{tab:peakbandwidth}.
The following five communication patterns have been simulated (note that all but
the first one are permutation patterns):
\begin{description}
\item[\rm\em Uniform:] Destinations are chosen independently and uniformly at
  random for each message.
\item[\rm\em Butterfly/FFT:] Messages are sent to nodes with Hamming distance
  one; the position of the flipped bit is increased at each message. This
  communication pattern is used, e.g., in FFT and bitonic sort.
\item[\rm\em Transposition:] The usual matrix transposition: $(i,j) \mapsto (j,i)$. It
  assumes the number of nodes $N$ to be square.
\item[\rm\em 3D Transposition:] 3D coordinates rotation: $(x,y,z) \mapsto (y,z,x)$ (this
  communication pattern has applications in medical imaging and 3D-FFT
  \cite{ElmoursyEtal08}).
\item[\rm\em Bit-Reverse:] Destination is the bit reverse of the source (e.g.,
  $0100 \mapsto 0010$). This pattern (which appears, e.g., in the FFT) is the
  worst case one for Dimension Order Routing on $n$-dimensional mesh-like
  networks.
\end{description}

We have simulated three-dimensional toroidal networks, therefore each router is
connected to six neighbors.  For each unidirectional link there are three VCs (one
adaptive and two escape ones) and the capacity of each VC is $C$ packets.
Routers also handle ACKs/NAKs to acknowledge packets acceptance, but we omit the
details in this paper due to space limitations.
Available packets are sent to the router for injection at rate $2.4\, \gamma_0$
(to quickly fill the router's queues when a message is generated).
We have simulated both a $8 \times 8 \times 8$ (512 nodes) and a $16 \times 16
\times 16$ (4096 nodes) torus network, with the parameters reported in
Table~\ref{tab:simu-param}.
The simulations have run in parallel on 160 threads on an Intel shared-memory
system (Xeon CPU E7-8850, 8 processors, 10 cores/processor, 2
hyper-threads/core) with a GNU/Linux operating system, and have taken 40 hours to
complete.
For each single packet flowing in the network the following information has been
tracked:
\begin{enumerate}
\item Number of traveled hops,
\item Lifetime from generation to consumption,
\item Whether and how derouted (WIDN vs OIDN).
\end{enumerate}

\begin{table}
  \caption{Parameters of the simulator.}
  \label{tab:simu-param}
  \centering
  \begin{tabular}{cl@{ }r@{ }l}
    \textbf{Parameter} & \textbf{Description} & \multicolumn{2}{l}{\textbf{Value}}\\
    \hline
    $S$ & Packet size & 512 & B\\
    $C$ & Queues capacity & 8 & packets\\
    $M$ & Message size & 96 & packets\\
    $\lambda_{\ii}$ & Internal latency & 80 & ns\\
    $B_{\ii}$ & Internal bandwidth & 64 & Gb/s\\
    $\lambda_{\ee}$ & External latency & 200 & ns\\
    $B_{\ee}$ & External bandwidth & 20 & Gb/s\\
    $k$ & Network size ($k \times k \times k$) & 8 & or 16 \\
    \hline
  \end{tabular}
\end{table}

\begin{figure*}
  \centering
  \begin{tabular}{ccc}
    \includegraphics[page=1,height=\plen]{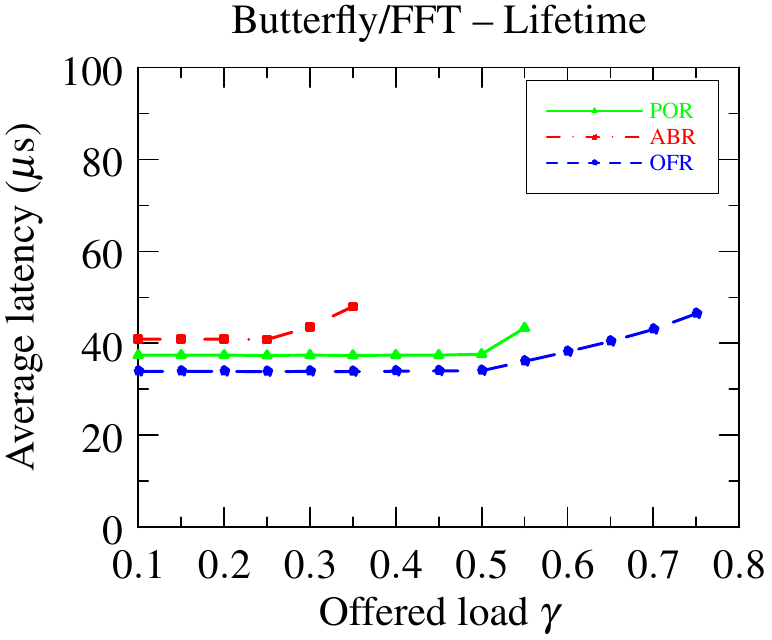} &
    \includegraphics[page=2,height=\plen]{images/plots-16.pdf} &
    \includegraphics[page=3,height=\plen]{images/plots-16.pdf} \\[5mm]
  \end{tabular}
  \begin{tabular}{cc}
    \includegraphics[page=4,height=\plen]{images/plots-16.pdf} &
    \includegraphics[page=5,height=\plen]{images/plots-16.pdf} 
  \end{tabular}
  \label{fig:plots-16}
  \caption{Simulation results: OFR vs.\ ABR vs.\ POR -- Average packet lifetime,
    $16 \times 16 \times 16$ network. Latencies not shown (e.g., ABR one for
    Butterfly/FFT and $\gamma=0.5$) are unbounded, the largest shown offered
    load for each pattern/routing policy is its maximum sustained throughput
    $\gamma^*$ (e.g., for Butterfly/FFT and ABR, $\gamma^*=0.35$). A summary of
    the maximum throughputs is given in Table~\ref{tab:peakbandwidth}, for both
    $k=16$ and $k=8$.}
\end{figure*}

\begin{table}
  \caption{maximum throughput $\gamma^*$ (in $\gamma_0$ units) on $k \times k
    \times k$ networks (except Transposition which has been tested on a $ 16 \times 8
    \times 8$ and on a $16 \times 16 \times 16 $ networks). Best results for
    each pattern/network size $k$ are highlighted in bold.}
  \label{tab:peakbandwidth}
  \centering
  \begin{tabular}{rr@{}cccccc}
    \textbf{Pattern} & & \multicolumn{2}{c}{\textbf{ABR}} & \multicolumn{2}{c}{\textbf{POR}} & \multicolumn{2}{c}{\textbf{OFR}} \\
    & $k=$ & \textbf{8} & \textbf{16} & \textbf{8} & \textbf{16} & \textbf{8} & \textbf{16} \\
    \hline
    Butterfly/FFT    && 0.30 & 0.35 & 0.55 & 0.55 & {\bf 0.60} & {\bf 0.75} \\
    Transposition    && 0.50 & 0.55 & {\bf 0.60} & 0.70 & 0.50 & {\bf 0.75} \\
    3D Transposition && 0.25 & 0.20 & {\bf 0.45} & {\bf 0.40}  & {\bf 0.45} & 0.35 \\ 
    Uniform          && 0.55 & 0.70 & {\bf 0.70} & {\bf 0.80} & {\bf 0.70} & {\bf 0.80} \\
    Bit-Reverse      && 0.35 & 0.40 & {\bf 0.60} & 0.50 & 0.50 & {\bf 0.60} \\
    \hline
  \end{tabular}
\end{table}

\begin{table}
  \caption{Ratio of derouted packets. OFR (OIDN and WIDN) vs.\ POR on a $16 \times 16
    \times 16$ network.}
  \label{tab:derouted}
  \centering
  \begin{tabular}{rccccccc}
    \textbf{Pattern} & \textbf{OFR/OIDN} & \textbf{OFR/WIDN} & \textbf{OFR (TOT)} & \textbf{POR}\\
    \hline
    Butterfly/FFT    & 0.31 & 0.10 & 0.41 & 0.35 \\
    Transposition    & 0.04 & 0.10 & 0.14 & 0.10 \\
    3D Transposition & 0.32 & 0.20 & 0.52 & 0.53 \\ 
    Uniform          & 0.06 & 0.12 & 0.18 & 0.13 \\
    Bit-Reverse      & 0.09 & 0.12 & 0.21 & 0.24 \\
    \hline
  \end{tabular}
\end{table}

\subsection{Results}
\label{sec:results}

The improvement achieved by OFR against ABR in maximum throughput is quite
remarkable, both for the $k=8$ and $k=16$ networks. The data reported in
Table~\ref{tab:peakbandwidth} for the $16 \times 16 \times 16$ network show that
there is an increment in the maximum sustainable bandwidth, which goes from
$+14\%$ for Uniform up to $+114\%$ for the extensively used Butterfly/FFT
communication pattern.

Since the POR policy already exploits non-minimal links, the improvement against
it is less conspicuous.
However, our policy increases its perform as the size of the network increases.
This is because the outflank derouting paths exploited by OFR come with a
constant dilation of up to $6\Delta$ ($=12$ hops in our simulations) which is
significant for small networks but becomes negligible as the average path length
increases.  Typical network size in top HPC systems \cite{top500} is some order
of magnitude larger than the networks, so the benefits of adopting OFR in these
systems are even stronger.
We can see empirical evidence of the improvement in performance for larger sizes
in our simulations.  Comparing OFR against POR we see that for the smaller $8
\times 8 \times 8$ network the maximum sustainable bandwidth is quite similar,
but POR wins over OFR in two cases, in one cases it loses and in two there is a
draw.
When moving to the larger $16 \times 16 \times 16$ network the situation is
reversed: OFR wins in three cases, loses in one and there is one draw,
registering an improvement against POR which is up to a $+36\%$ for the
Butterfly/FFT pattern.

In Table~\ref{tab:derouted} we show the ratio of packets which have been
derouted both by the OFR and the POR algorithms. The values remain pretty stable
within the range $\gamma\le \gamma^*$ (see, e.g., Fig.~\ref{fig:derouted}). One
might expect the ratio of derouted packets to be smaller at lower injection
rates than at higher ones, but since the traffic is generated in messages each
of $96$ packets, the routers' queues are rapidly filled even when the injection
rate is low.

\begin{figure}
  \centering
  \includegraphics[page=1,height=\plen]{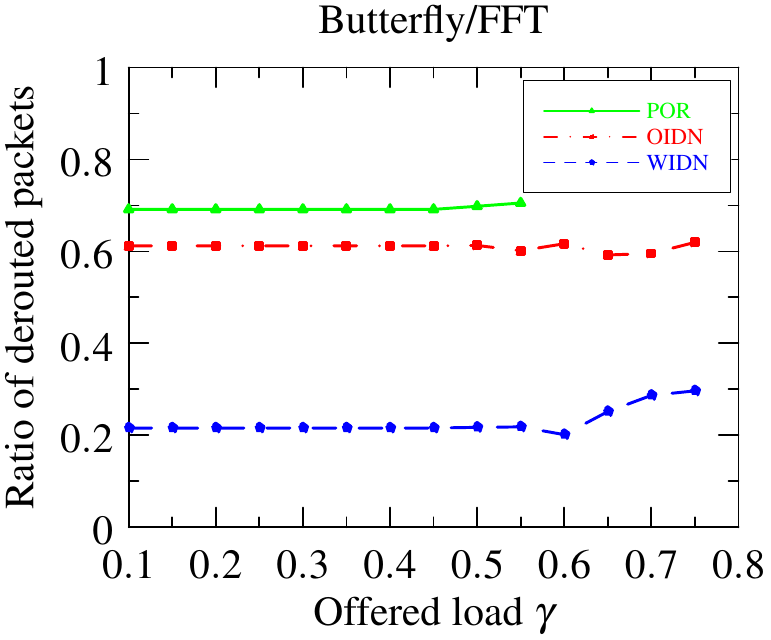}
  \caption{Ratio of derouted packets (OFR vs.\ POR) for the Butterfly/FFT
    pattern on a $16 \times 16 \times 16$ network.}
  \label{fig:derouted}
\end{figure}

\section{Conclusions and future work}
\label{sec:concl-future-work}

In this paper we have shown the potential of the OFR strategy, by exhibiting
significant improvements on network performance for some realistic communication
patterns. However, a few issues deserve further investigation:
\begin{itemize}
\item Before designing the hardware implementation of OFR, it would be
  interesting to strengthen its validation by testing the simulator on actual
  communication traces, rather than on synthetically generated ones.
\item Although the profit function (\ref{eq:profit}) adopted by OFR has proved
  to provide a good trade-off between congestion and path dilation, we would
  like to acquire a deeper understanding of the different effects of trading
  dilation for congestion, in order to theoretically justify the parameters and
  results obtained by empirical means.
\end{itemize}

\section*{Acknowledgment}
\label{sec:acknowledgment}

We wish to express our gratitude to Fabio Schifano, for providing useful
insights on the routing hardware and helping to discriminate implementable ideas
from unimplementable ones. We would also like to thank Dirk Pleiter for pointing
out some useful related work.

\bibliographystyle{acm}
\bibliography{aurouting}

\end{document}